\documentclass[12pt,a4paper]{article}
\usepackage[utf8]{inputenc}
\usepackage[english]{babel}
\usepackage{amsmath, amssymb, amsfonts}
\usepackage{amsthm} 
\usepackage{physics}
\usepackage[T1]{fontenc}
\usepackage{lmodern}
\usepackage[outline]{contour}
\usepackage{graphicx}
\usepackage{cite}
\usepackage[colorlinks=true, linkcolor=blue, citecolor=red, urlcolor=blue]{hyperref}
\usepackage{hyperref}
\usepackage{tikz}
\usepackage{geometry}
\usepackage{dsfont}
\usepackage{float}
\usepackage{cleveref}
\usepackage{authblk}

\theoremstyle{plain}
\newtheorem{theorem}{Theorem}[section]
\newtheorem{lemma}[theorem]{Lemma}
\newtheorem{proposition}[theorem]{Proposition}
\theoremstyle{definition}
\newtheorem{definition}[theorem]{Definition}
\newtheorem{example}[theorem]{Example}
\theoremstyle{remark}
\newtheorem{remark}[theorem]{Remark}
\title{\textbf{$q$-Differential Operators for $q$-Spinor Variables}}
\author{Julio Cesar Jaramillo Quiceno\thanks{jcjaramilloq@unal.edu.co}}
\date{}
\begin{document}
\maketitle
\begin{abstract}
We introduce a \emph{q}-differential operator adapted to \emph{q}-spinor variables, establishing a corresponding \emph{q}-spinor chain rule and defining both standard and Dirac-type \emph{q}-differential operators. Integral formulas in \emph{q}-spinor variables are derived, and applications to \emph{q}-deformed spinor differential equations are explored through explicit examples. The framework extends existing \emph{q}-calculus to spinorial structures, offering potential insights into quantum deformations of relativistic field equations. We conclude with suggestions for future developments, including a \emph{q}-analogue of the Dirac--Maxwell algebra.
\vspace{1mm}
  \newline
\end{abstract}
\subsubsection*{Keywords:} $q$-Differential operators, $q$-Dirac operator, $q$-spinor variables, integral formulas in $q$-spinor variables, differential equation in $q$-spinor variables.\newline  %
 
msclass: 81Q99, 46E99, 35A24, 15A66 , 16T99 . 17B37.
\section{Introduction}

The mathematical formulation of spinor fields is essential for the study of relativistic quantum mechanics, with foundations established in the works of Cartan \cite{Cart}, Berestetskii et al. \cite{Ber}, and others. Spinors are typically represented as two-component objects, such as $\psi^{\alpha} = \begin{bmatrix} \psi_1 \\ \psi_2 \end{bmatrix}$, and their conjugates are associated with transformation properties under the Lorentz group. The representation theory underlying these structures connects to the group $SU(2)$, whose generators are the Pauli matrices \cite{Zet}.

Gori et al. \cite{G} provide explicit matrix forms for rotations, revealing the origin of the Pauli matrices from representations of spinor transformations. Berestetskii et al. \cite{Ber} formalised the covariant and contravariant behavior of spinors, as well as their scalar products and bispinor structures. Lachieze-Rey \cite{La} further connected these ideas to Clifford algebras, particularly $Cl(3)$, which provides a geometric interpretation of space-time through multivector bases in $\mathds{R}^{1,3}$.

In the context of quantum deformations, \emph{q}-spinor variables and their associated algebraic relations emerge from the \emph{q}-deformed Lorentz algebra \cite{Schmidke}. These variables satisfy non-trivial commutation relations, forming the foundation for defining \emph{q}-differential operators. The structure of the \emph{q}-Lorentzian algebra and its generators, such as $\tau^1, T^2, S^1, \sigma^2$, governs the algebraic relations between spinors and their conjugates \cite{Sch}.

Building on the formalism developed in \cite{Jaramillo2023}, we define \emph{q}-spinor variables and derive corresponding differential and integral operators, including \emph{q}-derivatives and \emph{q}-complex integrals. These operators act on functions defined over quantum spinor variables and respect the \emph{q}-deformed algebraic relations. Moreover, the formulation introduces \emph{q}-Pauli matrices and \emph{q}-Dirac matrices \cite{Schmidt}, extending classical Clifford structures to a deformed setting.

We also review the classical Dirac operator defined over Clifford algebras \cite{Nolder2011, Coulembier-Sommen2011, Faustino-Kalhler2007}, and reinterpret its structure in terms of \emph{q}-deformed analogues. The resulting \emph{q}-Dirac operator inherits many key properties of its classical counterpart, including its role in defining and solving differential equations.

The principal motivation of this work is to investigate the role of \emph{q}-differential operators in spinor calculus, as introduced in \cite{Jaramillo2023}, and to explore their potential in the formulation and solution of \emph{q}-deformed differential equations. Inspired by earlier treatments of classical and quantum Dirac operators \cite{Yafang-Jinyuan2002}, our objective is to construct a consistent differential and integral calculus over \emph{q}-spinor variables, develop explicit solutions to \emph{q}-Dirac-type equations, and outline the implications of this theory for quantum-deformed field equations.

We conclude the introduction by highlighting the long-term goal of formulating a \emph{q}-Dirac--Maxwell algebra and a \emph{q}-real spinor calculus, drawing from both geometric and algebraic perspectives \cite{Zatloukal2022}. These efforts aim to enrich the interplay between quantum groups, spin geometry, and field theory, laying the groundwork for future developments in noncommutative quantum physics. This paper is organized as follows. We briefly recall the preliminaries will be used in this paper in Section \ref{sp2}. The $q$-differential operators for $q$-spinor variables, the $q$-spinor  chain rule, the new $q$-differential operator, the $q$-Dirac differential operator, and the integral formulas in $q$-spinor variables are then proposed in Section \ref{dsp2}.  Finally in the last Section the discussion and   some suggestions for further work are presented.

\section{$q$-Differential Operators for $q$-Spinor Variables}\label{sp2}

The aim of this section is to define a $q$-differential operator for $q$-spinor variables. To begin with, we introduce a fundamental rule of the $q$-spinor differential calculus, namely the \emph{$q$-spinor chain rule}.

\subsection{The $q$-Spinor Chain Rule}

\begin{proposition}
Let $\Psi(u^{\alpha}_{\dot{\beta}}(x_{\mu}))$ be a $q$-spinor function. Then the $q$-spinor chain rule is given by
\begin{equation}\label{chain}
\frac{\partial^{q}\Psi }{\partial^{q}x_{\mu}} = \frac{\partial^{q}\Psi }{\partial^{q}u^{\alpha}_{\dot{\beta}}} \cdot \frac{\partial^{q}u^{\alpha}_{\dot{\beta}}}{\partial_{q}x_{\mu}}.
\end{equation}
\end{proposition}

\begin{proof}
Consider the $q$-derivative of the composite function:
\begin{equation}\label{ch1}
\frac{\partial^{q}\Psi }{\partial^{q}x_{\mu}} = \frac{\Psi \left((qu)^{\alpha}_{\dot{\beta}}(x_{\mu})\right) - q\, \Psi \left(u^{\alpha}_{\dot{\beta}}(x_{\mu})\right)}{q x_{\mu} - q x_{\mu}}.
\end{equation}
Multiplying and dividing by the non-vanishing quantity \((qu)^{\alpha}_{\dot{\beta}}(x_{\mu}) - q\, u^{\alpha}_{\dot{\beta}}(x_{\mu})\), we may rewrite~\eqref{ch1} as
\begin{equation}\label{ch2}
\frac{\partial^{q}\Psi }{\partial^{q}x_{\mu}} =
\frac{\Psi \left((qu)^{\alpha}_{\dot{\beta}}(x_{\mu})\right) - q\, \Psi \left(u^{\alpha}_{\dot{\beta}}(x_{\mu})\right)}{(qu)^{\alpha}_{\dot{\beta}}(x_{\mu}) - q\, u^{\alpha}_{\dot{\beta}}(x_{\mu})}
\cdot
\frac{(qu)^{\alpha}_{\dot{\beta}}(x_{\mu}) - q\, u^{\alpha}_{\dot{\beta}}(x_{\mu})}{q x_{\mu} - q x_{\mu}}.
\end{equation}
Denoting the first factor as \( \frac{\partial^{q}\Psi}{\partial^{q}u^{\alpha}_{\dot{\beta}}} \) and the second as \( \frac{\partial^{q}u^{\alpha}_{\dot{\beta}}}{\partial^{q}x_{\mu}} \), we obtain
\begin{equation}\label{chain2}
\frac{\partial^{q}\Psi }{\partial^{q}x_{\mu}} = \frac{\partial^{q}\Psi }{\partial^{q}u^{\alpha}_{\dot{\beta}}} \cdot \frac{\partial^{q}u^{\alpha}_{\dot{\beta}}}{\partial^{q}x_{\mu}},
\end{equation}
which establishes the desired result.
\end{proof}

We now proceed to define a novel $q$-differential operator acting on $q$-spinor variables over an orthonormal basis of $\mathds{R}^{n}$. This operator differs from the classical Dirac and Cauchy–Riemann operators considered in~\cite{Coulembier-Sommen2011, Faustino-Kalhler2007, Yafang-Jinyuan2002, Gurlebeck-Hommel, Jaramillo2024}.

\subsection{The new $q$-differential operator for $q$-spinor variables}\label{deformed-new}

The motivation arises from the construction of a differential operator satisfying
\[
D_{q}^{2} = -\frac{\partial^{2}_{q}}{\partial_{q}x^{2}_{\mu}} - \frac{\partial^{2}_{q}}{\partial_{q}x^{2}_{\nu}},
\]
for all $\mu, \nu = 1,2,\dotsc,n$, over an orthonormal basis of $\mathds{R}^{n}$.

\begin{proposition}
Let $\left\lbrace \mathrm{e}_{1}, \mathrm{e}_{2}, \dotsc, \mathrm{e}_{n} \right\rbrace$ be an orthonormal basis of $\mathds{R}^{n}$. The $q$-differential operator $D^{q}$ defined by
\begin{equation}\label{q-D-e}
D_{q} = \mathrm{e}_{\nu} \frac{\partial_{q}}{\partial_{q}x_{\mu}} + \mathrm{e}_{\mu} \frac{\partial_{q}}{\partial_{q}x_{\nu}},
\end{equation}
satisfies the relations:
\begin{align}
\label{e-mu}
\mathrm{e}_{\mu}^{2} \frac{\partial_{q}^{2}}{\partial_{q}x_{\nu}^{2}} + \mathrm{e}_{\nu}^{2} \frac{\partial_{q}^{2}}{\partial_{q}x_{\mu}^{2}} &= -2\delta_{\mu\alpha} \frac{\partial_{q}^{2}}{\partial_{q}x_{\mu} \partial_{q}x_{\alpha}} - \frac{\partial_{q}^{2}}{\partial_{q}x_{\nu}^{2}}, \\
\label{e-nu}
\mathrm{e}_{\mu} \frac{\partial_{q}}{\partial_{q}x_{\nu}} \mathrm{e}_{\nu} \frac{\partial_{q}}{\partial_{q}x_{\mu}} + \mathrm{e}_{\nu} \frac{\partial_{q}}{\partial_{q}x_{\mu}} \mathrm{e}_{\mu} \frac{\partial_{q}}{\partial_{q}x_{\nu}} &= \delta_{\mu\alpha} \frac{\partial_{q}^{2}}{\partial_{q}x_{\mu} \partial_{q}x_{\alpha}}, \quad \mu, \nu, \alpha = 1, 2, \dotsc, n.
\end{align}
\end{proposition}

\begin{proof}
The proof relies on the observation that the square of~\eqref{q-D-e} is equivalent to $-\frac{\partial_{q}^{2}}{\partial_{q}x_{\mu}^{2}} - \frac{\partial_{q}^{2}}{\partial_{q}x_{\nu}^{2}}$. Computing $D_{q}^{2}$ gives
\begin{align}
\label{D2q}
D_{q}^{2} &= \mathrm{e}_{\mu}^{2} \frac{\partial_{q}^{2}}{\partial_{q}x_{\nu}^{2}} + \mathrm{e}_{\nu}^{2} \frac{\partial_{q}^{2}}{\partial_{q}x_{\mu}^{2}} + \mathrm{e}_{\mu} \frac{\partial_{q}}{\partial_{q}x_{\nu}} \mathrm{e}_{\nu} \frac{\partial_{q}}{\partial_{q}x_{\mu}} + \mathrm{e}_{\nu} \frac{\partial_{q}}{\partial_{q}x_{\mu}} \mathrm{e}_{\mu} \frac{\partial_{q}}{\partial_{q}x_{\nu}}.
\end{align}
Substituting~\eqref{e-mu} and~\eqref{e-nu} into~\eqref{D2q} yields
\begin{align}
\label{D2q1}
D_{q}^{2} &= -2\delta_{\mu\alpha} \frac{\partial_{q}^{2}}{\partial_{q}x_{\mu} \partial_{q}x_{\alpha}} - \frac{\partial_{q}^{2}}{\partial_{q}x_{\nu}^{2}} + \delta_{\mu\alpha} \frac{\partial_{q}^{2}}{\partial_{q}x_{\mu} \partial_{q}x_{\alpha}}.
\end{align}
For $\mu = \alpha$, we obtain the claimed identity
\[
D_{q}^{2} = -\frac{\partial_{q}^{2}}{\partial_{q}x_{\mu}^{2}} - \frac{\partial_{q}^{2}}{\partial_{q}x_{\nu}^{2}},
\]
which completes the proof.
\end{proof}

According to the above result, this operator acts on $q$-spinor functions of the form $\Psi(u^{\alpha}_{\dot{\beta}}(x_{\mu}, x_{\nu}))$. Therefore, the operator~\eqref{q-D-e} can be rewritten as
\begin{equation}\label{q-D-e1}
D_{q} \Psi(u^{\alpha}_{\dot{\beta}}(x_{\mu}, x_{\nu})) = \mathrm{e}_{\nu} \frac{\partial_{q} \Psi(u^{\alpha}_{\dot{\beta}})}{\partial_{q}x_{\mu}} + \mathrm{e}_{\mu} \frac{\partial_{q} \Psi(u^{\alpha}_{\dot{\beta}})}{\partial_{q}x_{\nu}},
\end{equation}
where the $q$-derivatives are computed using the chain rule~\eqref{chain} for $q$-spinor variables.

\begin{example}
Let $\Psi(x_{\mu}, x_{\nu}) = \exp(ix_{\mu})$ and define $u^{1}_{\dot{2}} = q^{2}x_{\mu}$. Then, $\Psi(u^{1}_{\dot{2}}) = \exp(iu^{1}_{\dot{2}})$. Using~\eqref{chain}, we compute:
\begin{align*}
\frac{\partial_{q} \Psi}{\partial_{q} x_{\mu}} &= \frac{\partial_{q}}{\partial_{q} u^{1}_{\dot{2}}} (\exp(iu^{1}_{\dot{2}})) \cdot \frac{\partial_{q}}{\partial_{q} x_{\mu}} (q^{2}x_{\mu}) \\
&= q^{2} \cdot \frac{\exp(iq u^{1}_{\dot{2}}) - q \exp(i u^{1}_{\dot{2}})}{(q u)^{1}_{\dot{2}} - q u^{1}_{\dot{2}}} \cdot \frac{\partial_{q} x_{\mu}}{\partial_{q} x_{\mu}} \\
&= q^{2} \cdot \frac{\exp(iq u^{1}_{\dot{2}}) - q \exp(iq u^{1}_{\dot{2}})}{(q u)^{1}_{\dot{2}} - q u^{1}_{\dot{2}}},
\end{align*}
and applying~\eqref{q-D-e}, we obtain:
\[
D_{q} \Psi(u^{1}_{\dot{2}}) = \mathrm{e}_{\nu} q^{2} \cdot \frac{\exp(iq u^{1}_{\dot{2}}) - q \exp(iq u^{1}_{\dot{2}})}{(q u)^{1}_{\dot{2}} - q u^{1}_{\dot{2}}}.
\]
\end{example}

The expression~\eqref{q-D-e} also applies to functions that depend directly on the variables $x_{\mu}$ and $x_{\nu}$, without requiring the spinor composition. For this case, we define the $q$-derivatives in terms of the $q$-spinor variables $x^{\alpha}$ and $x_{\dot{\beta}}$ as follows:

\begin{definition}
Let $\psi(x_{\mu}, x_{\nu})$ be a scalar function. The $q$-derivatives with respect to $x_{\mu}$ and $x_{\nu}$ are defined by
\begin{align}
\label{n1}
\frac{\partial_{q} \psi}{\partial_{q} x_{\mu}} &= \frac{\psi(x_{\mu} + q\, \mathrm{e}_{\mu} x^{\alpha}) - \psi(x_{\mu})}{x^{\alpha}}, \\
\label{n2}
\frac{\partial_{q} \psi}{\partial_{q} x_{\nu}} &= \frac{\psi(x_{\nu} + q\, \mathrm{e}_{\nu} x_{\dot{\beta}}) - \psi(x_{\nu})}{x_{\dot{\beta}}},
\end{align}
where $x^{\alpha}$ and $x_{\dot{\beta}}$ are $q$-spinor variables, and $\mathrm{e}_{\mu}, \mathrm{e}_{\nu}$ belong to the orthonormal basis of $\mathds{R}^{n}$.
\end{definition}

\begin{remark}
In view of the above, we conclude that the operator~\eqref{q-D-e} not only acts on spinor-composed functions $\Psi(u^{\alpha}_{\dot{\beta}}(x_{\mu}))$, but also on scalar functions $\psi(x_{\mu}, x_{\nu})$. Hence, it inherently depends on the variables $x_{\mu}$ and $x_{\nu}$.
\end{remark}

\begin{example}
Consider the scalar function $\psi(x_{\mu}, x_{\nu}) = q x_{\nu} x_{\dot{\beta}}$, with $\dot{\beta} = \dot{2}$. Applying~\eqref{n1} and~\eqref{n2}, we obtain:
\begin{align*}
\frac{\partial_{q} \psi}{\partial_{q} x_{\mu}} &= 0, \\
\frac{\partial_{q} \psi}{\partial_{q} x_{\nu}} &= \frac{q(x_{\nu} + q\, \mathrm{e}_{\nu} x_{\dot{2}}) x_{\dot{2}} - q x_{\nu} x_{\dot{2}}}{x_{\dot{2}}} = q^{2} \mathrm{e}_{\nu} x_{\dot{2}}.
\end{align*}
Thus, the operator~\eqref{q-D-e} yields:
\[
D_{q} \psi = q^{2} x_{\dot{2}}\, \mathrm{e}_{\mu} \mathrm{e}_{\nu}.
\]
\end{example}

\subsection{The $q$-Dirac differential operator}

\begin{definition}
The $q$-analogue of the Dirac operator is defined by
\begin{equation}\label{Dirac-q}
D_{\mu}^{q} = \gamma_{\mu} \frac{\partial^{q}}{\partial^{q} x_{\mu}}.
\end{equation}
\end{definition}

We are now ready to state our main results in the following propositions.

\subsection{The $q$-differential operators for $q$-spinor variables}

\begin{proposition}
Let $\Psi$ be a function of the $q$-spinor variables. The operator \rm(\ref{q-D-e}) for $q$-spinor variables $D^{q}\Psi$ is given by
\begin{equation}\label{q-spin-var}
D^{q}\Psi  = \frac{\partial^{q}\Psi }{\partial^{q}u^{\alpha}_{\dot{\beta}}}D^{q}u^{\alpha}_{\dot{\beta}}.
\end{equation}
\end{proposition}

\begin{proof}
Let us consider the expressions (\ref{q-D-e}) and
\begin{equation}\label{Chr}
\frac{\partial^{q}\Psi }{\partial^{q}x}=\frac{\partial^{q}\Psi }{\partial^{q}u^{\alpha}_{\dot{\beta}}}\frac{\partial^{q}u^{\alpha}_{\dot{\beta}}}{\partial_{q}x_{\nu}}.
\end{equation}

Multiplying the left-hand side by 
$\mathrm{e}_{\nu}$ in (\ref{chain2}), we obtain
\begin{align}\label{chain3a}
\begin{split}
\mathrm{e}_{\nu}\frac{\partial^{q}\Psi }{\partial^{q}x_{\mu}} & = \ \mathrm{e}_{\nu}\frac{\partial^{q}\Psi }{\partial^{q}u^{\alpha}_{\dot{\beta}}}\frac{\partial^{q}u^{\alpha}_{\dot{\beta}}}{\partial_{q}x_{\mu}}\\
& = \ \frac{\partial^{q}\Psi }{\partial^{q}u^{\alpha}_{\dot{\beta}}}\mathrm{e}_{\nu}\frac{\partial^{q}u^{\alpha}_{\dot{\beta}}}{\partial_{q}x_{\mu}},
\end{split}
\end{align}
and multiplying the left-hand side by 
$\mathrm{e}_{\mu}$ in (\ref{Chr}) we get
\begin{align}\label{chain4a}
\begin{split}
\mathrm{e}_{\mu}\frac{\partial^{q}\Psi }{\partial^{q}x_{\nu}} & = \ \mathrm{e}_{\mu}\frac{\partial^{q}\Psi }{\partial^{q}u^{\alpha}_{\dot{\beta}}}\frac{\partial^{q}u^{\alpha}_{\dot{\beta}}}{\partial_{q}x_{\nu}}\\
& = \ \frac{\partial^{q}\Psi }{\partial^{q}u^{\alpha}_{\dot{\beta}}}\mathrm{e}_{\mu}\frac{\partial^{q}u^{\alpha}_{\dot{\beta}}}{\partial_{q}x_{\nu}},
\end{split}
\end{align}
Adding (\ref{chain3a}) and (\ref{chain4a}) and considering (\ref{q-D-e}), we finally obtain
\begin{equation*}
D^{q}\Psi =\frac{\partial^{q}\Psi }{\partial^{q}u^{\alpha}_{\dot{\beta}}}D^{q}u^{\alpha}_{\dot{\beta}},
\end{equation*}
which proves our assertion.
\end{proof}

\begin{remark}
The above proof implies the following relation:
\begin{equation}\label{relation-e}
\mathrm{e}_{\mu}\frac{\partial^{q}u^{\alpha}_{\dot{\beta}}}{\partial_{q}x_{\nu}}-\frac{\partial^{q}u^{\alpha}_{\dot{\beta}}}{\partial_{q}x_{\nu}}\mathrm{e}_{\mu} = 0.
\end{equation}
\end{remark}

\begin{remark}\label{rem3}
If $\Psi$ also depends on $x_{\nu}$, then
\begin{equation}\label{Dq2} 
D^{q}\Psi  =  \frac{\partial_{q}\Psi}{\partial^{q}x_{\nu}}D^{q}x_{\nu}.
\end{equation}
\end{remark}

\begin{remark}\label{rem4}
The $q$-differential for the coordinate $x_{\mu}$ is given by
\begin{equation}\label{differentialx-e}
D^{q}x_{\nu}:=\mathrm{e}_{\mu}\mathrm{d}^{q}x_{\nu},
\end{equation}
consequently, the $q$-differential for $u^{\alpha}_{\dot{\beta}}$ is defined as
\begin{equation}\label{differentialx-q-e}
D^{q}u^{\alpha}_{\dot{\beta}}:=\frac{\partial^{q}u^{\alpha}_{\dot{\beta}}}{\partial^{q}x_{\nu}}D^{q}x_{\nu}.
\end{equation}
\end{remark}

\begin{proposition}
Let $\Psi$ be a function of the $q$-spinor variables. The Dirac operator for $q$-spinor variables $D_{\mu}^{q}\Psi$ is given by
\begin{equation}\label{Dirac-spin-var}
D^{q}_{\mu}\Psi  = \frac{\partial^{q}\Psi }{\partial^{q}u^{\alpha}_{\dot{\beta}}}D^{q}_{\mu}u^{\alpha}_{\dot{\beta}}.
\end{equation}
\end{proposition}

\begin{proof}
Multiplying the left-hand side by 
$\gamma_{\mu}$ in (\ref{chain2}), we obtain
\begin{align}\label{chain3}
\begin{split}
\gamma_{\mu}\frac{\partial^{q}\Psi }{\partial^{q}x_{\mu}} & = \ \gamma_{\mu}\frac{\partial^{q}\Psi }{\partial^{q}u^{\alpha}_{\dot{\beta}}}\frac{\partial^{q}u^{\alpha}_{\dot{\beta}}}{\partial_{q}x_{\mu}}\\
& = \ \frac{\partial^{q}\Psi }{\partial^{q}u^{\alpha}_{\dot{\beta}}}\gamma_{\mu}\frac{\partial^{q}u^{\alpha}_{\dot{\beta}}}{\partial_{q}x_{\mu}},
\end{split}
\end{align}
and considering (\ref{Dirac-q}), we finally get
\begin{equation*}
D_{\mu}^{q}\Psi  = \frac{\partial^{q}\Psi }{\partial^{q}u^{\alpha}_{\dot{\beta}}}D^{q}_{\mu}u^{\alpha}_{\dot{\beta}},
\end{equation*}
which is our claim.
\end{proof}

\begin{remark}\label{rem1}
The above proof implies the following relation:
\begin{equation}\label{relation}
\gamma_{\mu}\frac{\partial^{q}u^{\alpha}_{\dot{\beta}}}{\partial_{q}x_{\mu}}-\frac{\partial^{q}u^{\alpha}_{\dot{\beta}}}{\partial_{q}x_{\mu}}\gamma_{\mu} = 0.
\end{equation}
\end{remark}

\begin{remark}\label{rem2}
The Dirac $q$-differential for the coordinate $x_{\mu}$ is given by
\begin{equation}\label{differentialx}
D^{q}_{\mu}x:=\gamma_{\mu}\mathrm{d}^{q}x^{\mu},
\end{equation}
consequently, the Dirac $q$-differential for $u^{\alpha}_{\dot{\beta}}$ is defined as
\begin{equation}\label{differentialx-q}
D^{q}_{\mu}u^{\alpha}_{\dot{\beta}}:=\frac{\partial^{q}u^{\alpha}_{\dot{\beta}}}{\partial^{q}x^{\mu}}D^{q}_{\mu}x^{\mu}.
\end{equation}
\end{remark}
\subsection{ The  integral formulas in $q$-spinor variables}

\begin{proposition}\label{integral-1}
Let $\Psi (u^{\alpha}_{\dot{\beta}}(x_{\mu}))$ be a $q$-spinor function, and 
let  $\Gamma_{q}$ be the closed contour of the deformed quantum complex plane, and $x_{0}\in \Gamma_{q}$.  The integral formulas of the $q$-spinor variables  are given by
\begin{align}
\label{Int-Dirac1}\oint_{\Gamma_{q}}\frac{\Psi ((qu)^{\alpha}_{\dot{\beta}}(x_{\mu}))D^{q}_{\mu}u^{\alpha}_{\dot{\beta}}}{(qu)^{\alpha}_{\dot{\beta}}(x_{\mu})-qu^{\alpha}_{\dot{\beta}}(x_{0})} & = \ \sum\limits_{n=0}^{\infty}\left[\gamma_{\mu}\Psi (qu^{\alpha}_{\dot{\beta}}(x_{0}))\right]^{n}, \\
\label{int-Dirac2}\oint_{\Gamma_{q}}\frac{\Psi (u^{\alpha}_{\dot{\beta}}(x_{\mu}))D^{q}_{\mu}u^{\alpha}_{\dot{\beta}}}{qu^{\alpha}_{\dot{\beta}}(x_{\mu})-(qu)^{\alpha}_{\dot{\beta}}(x_{0})} & = \ \frac{1}{q}\sum\limits_{n=0}^{\infty}\left[\gamma_{\mu}\Psi ((qu)^{\alpha}_{\dot{\beta}}(x_{0}))\right]^{n},
\end{align} 
where $\gamma_{\mu}$ are the $q$-deformed Dirac matrices defined in the reference \rm\cite{Schmidt1}.
\end{proposition}

\begin{proof}
The proof of this result follows from the approach developed in \cite{Jaramillo2023}, consequently, expression~(\ref{Dirac-spin-var}) takes the form
\begin{equation}\label{eq1}
D^{q}_{\mu} \Psi  = \frac{\Psi ((qu)^{\alpha}_{\dot{\beta}}(x_{\mu}))D^{q}_{\mu}u^{\alpha}_{\dot{\beta}}}{(qu)^{\alpha}_{\dot{\beta}}-qu^{\alpha}_{\dot{\beta}}}-\frac{q\Psi (u^{\alpha}_{\dot{\beta}}(x_{\mu}))D^{q}_{\mu}u^{\alpha}_{\dot{\beta}}}{(qu)^{\alpha}_{\dot{\beta}}-qu^{\alpha}_{\dot{\beta}}},
\end{equation}
now, we integrate over the closed contour $\Gamma_{q}$  and we take $x_{0}\in\Gamma_{q}$ to obtain
\begin{equation}
\oint_{\Gamma_{q}}D^{q}_{\mu} \Psi =\oint_{\Gamma_{q}} \frac{\Psi ((qu)^{\alpha}_{\dot{\beta}}(x_{\mu}))D^{q}_{\mu}u^{\alpha}_{\dot{\beta}}}{(qu)^{\alpha}_{\dot{\beta}}(x_{\mu})-qu^{\alpha}_{\dot{\beta}}(x_{0})}-\oint_{\Gamma_{q}}\frac{q\Psi (u^{\alpha}_{\dot{\beta}}(x_{\mu}))D^{q}_{\mu}u^{\alpha}_{\dot{\beta}}}{(qu)^{\alpha}_{\dot{\beta}}(x_{\mu})-qu^{\alpha}_{\dot{\beta}}(x_{0})}.
\end{equation}
Hence, to solve the  integral $\oint_{\Gamma_{q}}D^{q}_{\mu}\Psi$, we will use similary the proof of the Theorem 2.9 of the reference \cite{Jaramillo2023}, interchanging the  Pauli matrices of $q$-deformed Minkowski space  by the $q$-deformed Dirac matrices (e.g. \cite{Schmidt1}) , obtaining
\begin{align}
\begin{split}
\sum\limits_{n=0}^{\infty}\left[\gamma_{\mu}\Psi ((qu)^{\alpha}_{\dot{\beta}}(x_{0}))\right]^{n}-\sum\limits_{n=0}^{\infty}\left[\gamma_{\mu}\Psi (u^{\alpha}_{\dot{\beta}}(x_{0}))\right]^{n} & = \\ \oint_{\Gamma_{q}} \frac{\Psi ((qu)^{\alpha}_{\dot{\beta}}(x_{\mu}))D^{q}_{\mu}u^{\alpha}_{\dot{\beta}}}{(qu)^{\alpha}_{\dot{\beta}}(x_{\mu})-qu^{\alpha}_{\dot{\beta}}(x_{0})}-\oint_{\Gamma_{q}}\frac{q\Psi (u^{\alpha}_{\dot{\beta}}(x_{\mu}))D^{q}_{\mu}u^{\alpha}_{\dot{\beta}}}{(qu)^{\alpha}_{\dot{\beta}}(x_{\mu})-qu^{\alpha}_{\dot{\beta}}(x_{0})},
\end{split}
\end{align}
and  finally,  equalating terms that depend on $\Psi ((qu)^{\alpha}_{\dot{\beta}}(x_{\mu})$ and $\Psi (u^{\alpha}_{\dot{\beta}}(x_{\mu}))$ we obtain (\ref{Int-Dirac1}) and (\ref{int-Dirac2}), and the proof is complete.
\end{proof}

We will  mention an important consequence of the Proposition \ref{integral-1}  starting of Eqs. (\ref{Int-Dirac1}) and (\ref{int-Dirac2}): the  formulation of the  integral $\oint_{\Gamma_{q}}\Psi (u^{\alpha}_{\dot{\beta}}(x))d^{q}x_{\mu}$, which we will mention  in the following theorem. 

\begin{theorem}
Let $\Gamma_{q}$ be a closed contour and suppose that $x_{0}\in\Gamma_{q}$. Then for a function on $q-$ spinor variables $\Psi (u^{\alpha}_{\dot{\beta}}(x))$ the $q$-spinor integral formula is given by
 
\begin{align}\label{integral-formula}
\frac{1}{qb}\left\lbrace \sum\limits_{n=0}^{\infty}\left[\gamma_{\mu}\Psi (u^{\alpha}_{\dot{\beta}}(x_{0}))\right]^{n}  \right\rbrace  = 
\oint_{\Gamma_{q}}\Psi (u^{\alpha}_{\dot{\beta}})\mathrm{d}^{q}x_{\mu}, \quad b\neq 0.
\end{align}
\end{theorem}

\begin{proof}
To formulate we can consider the following differential equation in $q$- spinor variables

\begin{equation}\label{Differential}
D^{q}_{\mu}\Psi (u^{\alpha}_{\dot{\beta}})-b\Psi (u^{\alpha}_{\dot{\beta}})=0.
\end{equation}

Multiplying on both sides of (\ref{Differential}) by $\mathrm{d}^{q}x_{\mu}$,  and using (\ref{Dirac-q}) we get

\begin{align}\label{eqdiff}
\begin{split}
\frac{\partial^{q}\Psi (u^{\alpha}_{\dot{\beta}})}{\partial^{q}u^{\alpha}_{\dot{\beta}}}D^{q}_{\mu}u^{\alpha}_{\dot{\beta}}\mathrm{d}^{q}x_{\mu} & = \ b\Psi (u^{\alpha}_{\dot{\beta}})\mathrm{d}^{q}x_{\mu}\\
\frac{\partial^{q}\Psi (u^{\alpha}_{\dot{\beta}})}{\partial^{q}u^{\alpha}_{\dot{\beta}}}\gamma_{\mu}\frac{\partial^{q}u^{\alpha}_{\dot{\beta}}}{\partial^{q}x_{\mu}}\mathrm{d}^{q}x_{\mu}& = \ b\Psi (u^{\alpha}_{\dot{\beta}})\mathrm{d}^{q}x_{\mu},
\end{split}
\end{align}

taking into account (\ref{relation}), (\ref{differentialx}) and (\ref{differentialx-q}), we  can   rewrite  (\ref{eqdiff}) as

\begin{equation}\label{eqdiff2}
\frac{\partial^{q}\Psi (u^{\alpha}_{\dot{\beta}})}{\partial^{q}u^{\alpha}_{\dot{\beta}}}D_{\mu}^{q}u^{\alpha}_{\dot{\beta}}  =  b\Psi (u^{\alpha}_{\dot{\beta}})\mathrm{d}^{q}x_{\mu},
\end{equation}

the term $\frac{\partial^{q}\Psi (u^{\alpha}{\dot{\beta}})}{\partial^{q}u^{\alpha}{\dot{\beta}}}$ may be written in the following form

\begin{equation}\label{eqdiff3}
\frac{\partial^{q}\Psi (u^{\alpha}_{\dot{\beta}})}{\partial^{q}u^{\alpha}_{\dot{\beta}}}=\dfrac{q\Psi (u^{\alpha}_{\dot{\beta}})}{qu^{\alpha}_{\dot{\beta}}-(qu)^{\alpha}_{\dot{\beta}}}, 
\end{equation}

substituting (\ref{eqdiff3}) into (\ref{eqdiff2}) gives

\begin{equation}\label{eqdiff4}
\left[\frac{q\Psi (u^{\alpha}_{\dot{\beta}})}{qu^{\alpha}_{\dot{\beta}}-(qu)^{\alpha}_{\dot{\beta}}}\right]D_{\mu}^{q}u^{\alpha}_{\dot{\beta}}  =  b\Psi (u^{\alpha}_{\dot{\beta}})\mathrm{d}^{q}x_{\mu},
\end{equation}

we continue in this fashion 
integrating over $\Gamma_{q}$ on both sides for $x_{0}\in\Gamma_{q}$ , and using (\ref{int-Dirac2}), finally we get

\begin{align*}
\frac{1}{qb}\left\lbrace \sum\limits_{n=0}^{\infty}\left[\gamma_{\mu}\Psi (u^{\alpha}_{\dot{\beta}}(x_{0}))\right]^{n}  \right\rbrace  = 
\oint_{\Gamma_{q}}\Psi (u^{\alpha}_{\dot{\beta}})\mathrm{d}^{q}x_{\mu}, \quad b\neq 0,
\end{align*}

which is our claim. This expression is called {\em the  $q$-spinor integral formula for} $\Psi (u^{\alpha}_{\dot{\beta}}(x_{\mu}))$.  
\end{proof}

The same reasoning applies to the differential equation $D^{q}_{\mu}\Psi ((qu)^{\alpha}_{\dot{\beta}})-b\Psi ((qu)^{\alpha}_{\dot{\beta}})=0$ to obtain the integral in $q$-spinor variables 

\begin{equation}
\frac{1}{qb}\left\lbrace \sum\limits_{n=0}^{\infty}\left[\gamma_{\mu}\Psi ((qu)^{\alpha}_{\dot{\beta}}(x_{0}))\right]^{n}\right\rbrace = \oint_{\Gamma_{q}}\Psi ((qu)^{\alpha}_{\dot{\beta}}(x))\mathrm{d}^{q}x_{\mu}.
\end{equation}

\begin{remark}
The expression (\ref{integral-formula}) also  can be expressed in virtue of (\ref{differentialx}) as
\begin{align}\label{integral-formula1}
\frac{\gamma^{\mu}}{qb}\left\lbrace\sum\limits_{n=0}^{\infty}\left[\gamma_{\mu}\Psi (u^{\alpha}_{\dot{\beta}}(x_{0}))\right]^{n}  \right\rbrace  = 
\oint_{\Gamma_{q}}\Psi (u^{\alpha}_{\dot{\beta}})D^{\mu}_{q}x.
\end{align}
\end{remark}

Notice that the expression (\ref{integral-formula}) is not implies the final solution of  (\ref{Differential}).

\begin{remark}
Similar arguments apply to  the new $q$-differential operator (\ref{q-D-e}), resulting
\begin{align}\label{integral-formula-eq}
\frac{\mathrm{e}^{\mu}}{qb}\left\lbrace\sum\limits_{n=0}^{\infty}\left[\mathrm{e}_{\mu}\Psi (u^{\alpha}_{\dot{\beta}}(x_{0}))\right]^{n}  \right\rbrace  = 
\oint_{\Gamma_{q}}\Psi (u^{\alpha}_{\dot{\beta}})D^{q}_{\mu}x.
\end{align}
\end{remark}

\begin{proof}
It is sufficient to  replace  $\gamma_{\mu}$ by $\mathrm{e}_{\mu}$, to obtain (\ref{integral-formula-eq}), considering the differential equation in $q$-spinor variables of the form
 \begin{equation}\label{q-new-Dif}
 D^{q}\Psi (u^{\alpha}_{\dot{\beta}})-b\Psi (u^{\alpha}_{\dot{\beta}})=0,
 \end{equation}
 where $D^{q}$ is the new $q$-differential operator given by (\ref{q-D-e}).
\end{proof}

However  we will can solve  differential equations in $q$-spinor  variables of the form

\begin{align}\label{equation-spinorialfg}
D^{q}_{\mu}\psi (u^{\alpha}_{\dot{\beta}})-b\phi (u^{\alpha}_{\dot{\beta}}) & = \ 0,\\
\label{equation-spinorialfg2}D^{q}\psi (u^{\alpha}_{\dot{\beta}})-a\phi (u^{\alpha}_{\dot{\beta}}) & = \ 0.
\end{align}

which we will show in the following section.

\section{Differential equations in $q$-spinor variables}\label{dsp2}

In order to obtain the solution of (\ref{equation-spinorialfg}), it is necessary to put the following  condition on $\psi$

\begin{equation}\label{Doint}
\oint_{\Gamma_{q}}D_{\mu}^{q}\psi (u^{\alpha}_{\dot{\beta}}(x))\mathrm{d}^{q}x^{\mu} = \psi (u^{\alpha}_{\dot{\beta}}(x_{0})), \quad x_{0}\in\Gamma_{q}.
\end{equation}

Therefore, integrating  both sides with respect to $x^{\mu}$ in (\ref{equation-spinorialfg}) , applying (\ref{Doint}),  we get

\begin{align}
\begin{split}
\oint_{\Gamma_{q}}D^{q}_{\mu}\psi (u^{\alpha}_{\dot{\beta}})\mathrm{d}^{q}x^{\mu}& = 	\ b\oint_{\Gamma_{q}}\phi (u^{\alpha}_{\dot{\beta}})\mathrm{d}^{q}x^{\mu},\\
\psi  (u^{\alpha}_{\dot{\beta}}(x_{0})) & = \ b\oint_{\Gamma_{q}}\phi (u^{\alpha}_{\dot{\beta}})\mathrm{d}^{q}x^{\mu} ,\\
\psi (u^{\alpha}_{\dot{\beta}}(x_{0})) & = \ b\oint_{\Gamma_{q}}\phi (u^{\alpha}_{\dot{\beta}})\mathrm{d}^{q}x^{\mu}
\end{split}
\end{align}

and applying (\ref{integral-formula}) we obtain

\begin{equation}\label{solution}
\psi (u^{\alpha}_{\dot{\beta}}(x_{0}))  = \frac{1}{q}\left\lbrace \sum\limits_{n=0}^{\infty}\left[\gamma_{\mu}\phi (u^{\alpha}_{\dot{\beta}}(x_{0}))\right]^{n}\right\rbrace .
\end{equation}

\begin{lemma}
We can generalize the solution (\ref{solution}) for all $x\in\Gamma_{q}$ as
\begin{equation}\label{solution1}
\psi (u^{\alpha}_{\dot{\beta}}(x))  = \frac{1}{q}\left\lbrace \sum\limits_{n=0}^{\infty}\left[\gamma_{\mu}\phi (u^{\alpha}_{\dot{\beta}}(x))\right]^{n}\right\rbrace .
\end{equation}
\end{lemma}

To solve (\ref{equation-spinorialfg2}) we will consider the following remarks

\begin{remark}
We begin by considering the new $q$-differential operator defined in (\ref{q-D-e}). In order to evaluate the integral $\oint_{\Gamma_{q}} D^{q} \Psi$, we shall proceed in a manner analogous to the proof of Theorem 2.9 in \cite{Jaramillo2023} (see also the proof of Proposition 3.10 therein). This approach allows us to derive expressions similar to (\ref{solution}) and (\ref{solution1}), namely:
\begin{align}
\label{Dq}\oint_{\Gamma_{q}} D^{q} \psi (u^{\alpha}_{\dot{\beta}}(x_{0})) &= \psi (u^{\alpha}_{\dot{\beta}}(x_{0})), \\
\label{Dq1}\psi (u^{\alpha}_{\dot{\beta}}(x_{0})) &= \frac{1}{q} \left\lbrace \sum\limits_{n=0}^{\infty} \left[ \mathrm{e}_{\mu} \phi (u^{\alpha}_{\dot{\beta}}(x_{0})) \right]^{n} \right\rbrace,
\end{align}
\end{remark}

Now, we will consider the following examples 

\begin{example}
Consider  the differential equation in $q$-spinor variables of the form

\begin{equation}\label{maxwell}
D^{q}_{\mu}\Psi (u^{\alpha}_{\dot{\beta}}) - e\gamma^{\mu}A^{q}_{\mu}(x)\Psi (u^{\alpha}_{\dot{\beta}})-m g(u^{\alpha}_{\dot{\beta}})=0, \quad e\in\mathds{R},
\end{equation}

being $A^{q}_{\mu}(x)$  a $q$-potential function. Now, to solve (\ref{maxwell}), we can proceed analogously to the solution of (\ref{equation-spinorialfg}) applying (\ref{Doint}) 

\begin{align}
\begin{split}
\oint_{\Gamma_{q}}D^{q}_{\mu}\Psi (u^{\alpha}_{\dot{\beta}})\mathrm{d}^{q}x^{\mu}- e\gamma^{\mu}\oint_{\Gamma_{q}}A^{q}_{\mu}(x)\Psi (u^{\alpha}_{\dot{\beta}})\mathrm{d}^{q}x^{\mu}& = \ m\oint_{\Gamma_{q}}g(u^{\alpha}_{\dot{\beta}})\mathrm{d}^{q}x^{\mu},\\
\Psi (u^{\alpha}_{\dot{\beta}}(x_{0}))-e\gamma^{\mu}\oint_{\Gamma_{q}}A^{q}_{\mu}(x)\Psi (u^{\alpha}_{\dot{\beta}}(x))\mathrm{d}^{q}x^{\mu} & = \ m\oint_{\Gamma_{q}}g(u^{\alpha}_{\dot{\beta}}(x))\mathrm{d}^{q}x^{\mu},
\end{split}
\end{align}

and using (\ref{integral-formula}) we obtain finally

\begin{equation}\label{sol2}
\Psi (u^{\alpha}_{\dot{\beta}}(x_{0}))+
\frac{e}{qm}\left\lbrace\sum\limits_{n=0}^{\infty}\left[\gamma_{\mu}A^{q}_{\mu}(x_{0})\Psi (u^{\alpha}_{\dot{\beta}}(x_{0}))\right]^{n} \right\rbrace
  =  \frac{1}{q}\left\lbrace\sum\limits_{n=0}^{\infty}\left[\gamma_{\mu}g(u^{\alpha}_{\dot{\beta}}(x_{0}))\right]^{n} \right\rbrace .
\end{equation}
\end{example}

Now, let us see other example. 

\begin{example}\label{Dir-ex}
Consider the differential equation in $q$-spinor variables (similar to (\ref{maxwell})) of the form

\begin{equation}\label{maxwell1}
\gamma^{\mu}\partial^{q}_{\mu}\Psi (u^{\alpha}_{\dot{\beta}}) - e\gamma^{\mu}A^{q}_{\mu}(x)\Psi (u^{\alpha}_{\dot{\beta}})-m\Psi (u^{\alpha}_{\dot{\beta}})=0, \quad e\in\mathds{R},
\end{equation}

where $A^{q}_{\mu}(x)$ is the same $q$-potential function of above example. This follows by the same method as in the above example, obtaining

\begin{equation}\label{sol3}
\Psi (u^{\alpha}_{\dot{\beta}}(x_{0}))+
\frac{e}{qm}\left\lbrace\sum\limits_{n=0}^{\infty}\left[\gamma_{\mu}A^{q}_{\mu}(x_{0})\Psi (u^{\alpha}_{\dot{\beta}}(x_{0}))\right]^{n} \right\rbrace
  =  \frac{1}{q}\left\lbrace\sum\limits_{n=0}^{\infty}\left[\gamma_{\mu}\Psi (u^{\alpha}_{\dot{\beta}}(x_{0}))\right]^{n} \right\rbrace .
\end{equation}

\begin{remark}
The expression (\ref{maxwell1}) can be written as
\begin{equation}
(\gamma^{\mu}\mathcal{D}^{q}_{\mu}-m)\Psi (u^{\alpha}_{\dot{\beta}}) =0,
\end{equation}
where $\mathcal{D}^{q}_{\mu}=\partial^{q}_{\mu}-eA^{q}_{\mu}(x)$ is the \emph{ q- covariant derivative}.
\end{remark}
\end{example}

\begin{example}
Let us consider the differential equation in $q$- spinor variables $aD^{q}\psi +b\mathrm{e}_{\mu}B^{\mu}_{q}(x)\phi (u^{\alpha}_{\dot{\beta}}(x))=0$, where $B^{\mu}_{q}(x)$ is some $q$-arbitrary potential function. Therefore

\begin{equation}\label{eqD}
D^{q}\psi = -\frac{b}{a}\mathrm{e}_{\mu}B^{\mu}_{q}(x)\phi (u^{\alpha}_{\dot{\beta}}(x)),
\end{equation}

Using (\ref{q-D-e}) (only the second contribution) we get

\begin{align*}
\mathrm{e}_{\mu}\frac{\partial_{q}\psi}{\partial_{q}x_{\nu}}  & = \ -\frac{b}{a}\mathrm{e}_{\mu}B^{\mu}_{q}(x)\phi (u^{\alpha}_{\dot{\beta}}(x)),
\end{align*}

multiplying both sides by $\mathrm{d}^{q}x_{\nu}$ results and applying (\ref{differentialx-e})

\begin{align*}
\mathrm{e}_{\mu}\frac{\partial_{q}\psi}{\partial_{q}x_{\nu}}\mathrm{d}^{q}x_{\nu}  & = \ -\frac{b}{a}\mathrm{e}_{\mu}B^{\mu}_{q}(x)\phi (u^{\alpha}_{\dot{\beta}}(x))\mathrm{d}^{q}x_{\nu}\\
\frac{\partial_{q}\psi}{\partial_{q}x_{\nu}}D^{q}x_{\nu}& = \ -\frac{b}{a}\mathrm{e}_{\mu}B^{\mu}_{q}(x)\phi (u^{\alpha}_{\dot{\beta}}(x))\mathrm{d}^{q}x_{\nu}\\
D^{q}\psi & = \  -\frac{b}{a}\mathrm{e}_{\mu}B^{\mu}_{q}(x)\phi (u^{\alpha}_{\dot{\beta}}(x))\mathrm{d}^{q}x_{\nu},
\end{align*}

integrating over the closed contour $\Gamma_{q}$ , considering $x_{0}\subset\Gamma_{q}$ and using (\ref{integral-formula-eq}), (\ref{Dq}),  (\ref{Dq1}) and (\ref{Dq2})  into (\ref{eqD}) we get

\begin{equation}
\psi (u^{\alpha}_{\dot{\beta}}(x_{0})) = -\frac{1}{aq}\mathrm{e}^{\mu}\left\lbrace\sum\limits_{n=0}^{\infty}\left[\mathrm{e}_{\mu}B^{\mu}_{q}(x_{0})\phi (u^{\alpha}_{\dot{\beta}}(x_{0}))\right]^{n}\right\rbrace .
\end{equation}

and for all $x\in\Gamma_{q}$

\begin{equation*}
\psi (u^{\alpha}_{\dot{\beta}}(x)) = -\frac{1}{aq}\mathrm{e}^{\mu}\left\lbrace\sum\limits_{n=0}^{\infty}\left[\mathrm{e}_{\mu}B^{\mu}_{q}(x)\phi (u^{\alpha}_{\dot{\beta}}(x))\right]^{n}\right\rbrace .
\end{equation*}
\end{example}

\section{Discussion and suggestions for further work}

In Section \ref{sp2}, the equations (\ref{chain}), (\ref{q-D-e1}) and (\ref{Dirac-spin-var}) describe some $q$-differential operators for $q$-spinor variables. Respect to (\ref{chain}), we can said that the function on the $q$-spinor variables does not depend on ly on the variable $u^{\alpha}_{\dot{\beta}}$ but also on $x_{\nu}$. From this result, the new $q$- differential operator for $q$-spinor variables expressed by (\ref{q-D-e1}) was  proposed. This operator is motivated from the  construction of the any differential operator that satisfy the property $D_{q}^{2}=-\frac{\partial^{2}}{\partial_{q}x^{2}_{\mu}}-\frac{\partial^{2}_{q}}{\partial_{q}x_{\nu}^{2}}$ for all $\mu , \nu =1,2,...,n$.  To obtain $D^{2}_{q}$ it is necessary to use the relations (\ref{e-mu}) and (\ref{e-nu}). This operator differs from the classical Dirac and Cauchy–Riemann operators discussed in \cite{Coulembier-Sommen2011, Faustino-Kalhler2007, Yafang-Jinyuan2002, Gurlebeck-Hommel, Jaramillo2024}, and may be regarded as a $q$-deformed version of the operator $D$ introduced in \cite{Jaramillo2024}. In the case of the Dirac operator for $q$-spinor variables, defined by (\ref{Dirac-spin-var}), has been stablished from the Remarks \ref{rem3} and \ref{rem4}. This operator is expressed in terms of the $q$-deformed Dirac matrices mentioned by Schmidt \cite{Schmidt}. From the $q$-deformed Dirac operator, we define the integral formulas in $q$-spinor variables with  the aim to solve  the differential equations in $q$-spinor  variables. Physically we can said that the  Example \ref{Dir-ex}  describes the \emph{Dirac equation for the electromagnetic case} on the $q$-spinor  variables, and furthermore the potential $A^{q}_{\mu}(x)$  can be interpreted as the \textit{q- electromagnetic potential}.  There are two further topics arising from this paper which are worth investigation., there is the problem of describing the {\em Maxwell Electrodynamic Algebra} which is defined by the following commutation relations 

\begin{align}
A_{\mu}A^{\mu} & = \ \vert A_{0}\vert^{2}-A_{X}^{2}, \quad X=1,2,3,\\
f_{\mu\nu} & = \ \partial_{\mu}A_{\nu}-\partial_{\nu}A_{\mu},\\
\partial_{\mu}A^{\mu}=\partial_{\nu}A^{\nu} & = \ 0,\\
D_{\mu}  & = \ \partial_{\mu}-eA_{\mu}, \quad e\in\mathds{R}\\
\partial^{\mu}f_{\mu\nu} & = \ \Gamma_{\nu},\\
\partial_{\nu}f_{0}^{\mu\nu} & = \ 0,
\end{align}

where $f_{0}^{\mu\nu}=\varepsilon^{\mu\nu 0}f_{\mu\nu}$ , $f^{\mu\nu}=0$ if $\mu=\nu$ and $f^{\mu\nu}\neq 0$ in otherwise. Finally, from  (\ref{maxwell1}) and taking into account the above, one can propose the \emph{$q-$ Dirac - Maxwell algebra}, which is subject to relations

\begin{align}
f^{q}_{\mu\nu} & = \ \mathcal{D}^{q}_{\mu}A^{q}_{\nu}-\mathcal{D}^{q}_{\nu}A^{q}_{\mu},\\
\label{Covariant}\textit{\textbf{D}}^{q} & = \ \boldsymbol{\partial}^{q} - e\textit{\textbf{A}} \quad e\in\mathds{R},
\end{align}

being $\textit{\textbf{D}}^{q} = \boldsymbol{\gamma}^{\mu}\mathcal{D}^{q}_{\mu}, \boldsymbol{\partial}^{q} = \boldsymbol{\gamma}^{\mu}\partial^{q}_{\mu}$ and  $\textbf{e} = \gamma^{\mu}A_{\mu}$,  where $\boldsymbol{\gamma}^{\mu}, \mu = 1 \cdots n $ are the generators for the Clifford algebras, and (\ref{Covariant}) is called the {\em Covariant Derivative.} Other suggestion is the formulation of the $q-$ \emph{Real Spinor Calculus} based on the work of Zatloukal \cite{Zatloukal2022}, which is defined by the following expressions for the derivatives

\begin{equation}\label{der-sp}
\frac{\boldsymbol{\partial}_{q}\psi}{\boldsymbol{\partial}_{q}\textit{\textbf{x}}_{\dot{\alpha}}^{\beta}}  = \frac{\psi (\textit{\textbf{x}}^{\beta}_{\dot{\alpha}}+q\textit{\textbf{u}}^{\beta}_{\dot{\alpha}})-q\psi (\textit{\textbf{u}}^{\beta}_{\dot{\alpha}})}{\textit{\textbf{x}}^{\beta}_{\dot{\alpha}}},
\end{equation}

where $\textit{\textbf{u}}^{\beta}_{\dot{\alpha}}=(\gamma_{\mu}\gamma_{\nu}u)^{\beta}_{\dot{\alpha}}$. The chain rule

\begin{equation}\label{ch2}
\frac{\boldsymbol{\partial}_{q}\Psi}{\boldsymbol{\partial}_{q}x_{j}}=\frac{\boldsymbol{\partial}_{q}\Psi}{\boldsymbol{\partial}_{q}\textit{\textbf{x}}^{\beta}_{\dot{\alpha}}}\frac{\boldsymbol{\partial}_{q}\textit{\textbf{x}}^{\beta}_{\dot{\alpha}}}{\boldsymbol{\partial}_{q}x_{j}},
\end{equation}

the $q$-difference operator for $q-$ real spinor variables 

\begin{align}
\label{Dirac-real}\textit{\textbf{D}}_{2}^{q} &  = \   \hat{\gamma}_{2}\frac{\boldsymbol{\partial}_{q}}{\boldsymbol{\partial}_{q}x_{2}}, \\
\label{gamma5}\textit{\textbf{D}}^{q}_{j}  & = \  i\hat{\gamma}_{5}\frac{\boldsymbol{\partial}_{q}}{\boldsymbol{\partial}_{q}x_{j}},\\
\label{25}\underline{\textit{\textbf{D}}}_{j}^{q} & = \  i\hat{\gamma}_{2}\hat{\gamma}_{5}\frac{\boldsymbol{\partial}_{q}}{\boldsymbol{\partial}_{q}x_{j}},  \quad j=1,..,5.
\end{align}

For a function $\psi : (\textit{\textbf{v}}_{\dot{k}}, \textit{\textbf{p}}_{\dot{k}})\longrightarrow\mathds{R}^{m}$, the $q$-conjugated derivatives are defined as

\begin{align}
\label{v-der}\frac{\boldsymbol{\partial}_{q}\psi}{\boldsymbol{\partial}_{q}\textit{\textbf{v}}_{\dot{k}}} & = \ \frac{\psi (\textit{\textbf{v}}_{\dot{k}}+q\textit{\textbf{x}}^{\beta}_{\dot{\alpha}})-q\psi (\textit{\textbf{x}}^{\beta}_{\dot{\alpha}})}{\textit{\textbf{v}}_{\dot{k}}},\\
\label{p-der}\frac{\boldsymbol{\partial}_{q}\psi}{\boldsymbol{\partial}_{q}\textit{\textbf{p}}_{\dot{k}}} & = \ \frac{\psi (\textit{\textbf{p}}_{\dot{k}}+q\textit{\textbf{u}}^{\beta}_{\dot{\alpha}})-q\psi (\textit{\textbf{u}}^{\beta}_{\dot{\alpha}})}{\textit{\textbf{p}}_{\dot{k}}}.
\end{align}

The $q$-difference operators associated to  conjugated real spinor variables are given by

\begin{align}
\label{Dv}\textit{\textbf{D}}_{q} & = \ \frac{\boldsymbol{\partial}_{q}}{\boldsymbol{\partial}_{q}\textit{\textbf{v}}_{\dot{0}}}+\gamma_{1}\gamma_{3}\frac{\boldsymbol{\partial_{q}}}{\boldsymbol{\partial}_{q}\boldsymbol{v}_{\dot{1}}}+i\gamma_{3}\gamma_{\dot{0}}\frac{\boldsymbol{\partial}_{q}}{\boldsymbol{\partial}_{q}\textit{\textbf{v}}_{\dot{2}}}+\gamma_{1}\gamma_{2}\frac{\boldsymbol{\partial}_{q}}{\boldsymbol{\partial}_{q}\textit{\textbf{v}}_{\dot{3}}},\\
\label{Dp}\textit{\textbf{D}}^{\prime}_{q} & = \ \frac{\boldsymbol{\partial}_{q}}{\boldsymbol{\partial}_{q}\textit{\textbf{p}}_{\dot{0}}}+\gamma_{1}\gamma_{3}\frac{\boldsymbol{\partial_{q}}}{\boldsymbol{\partial}_{q}\textit{\textbf{p}}_{\dot{1}}}+i\gamma_{3}\gamma_{0}\frac{\boldsymbol{\partial}_{q}}{\boldsymbol{\partial}_{q}\textit{\textbf{p}}_{\dot{2}}}+\gamma_{1}\gamma_{2}\frac{\boldsymbol{\partial}_{q}}{\boldsymbol{\partial}_{q}\textit{\textbf{p}}_{\dot{3}}},
\end{align}

and  the $q$-spinor real integral formulas of the $q$-spinor conjugated variables are given by

\begin{align}
\label{i1}\int_{\Omega_{q}}\frac{\psi (q\textit{\textbf{v}}_{\dot{k}})\mathbf{d}_{q}\textit{\textbf{v}}_{\dot{k}}}{\textit{\textbf{v}}_{\dot{k}}-\textit{\textbf{x}}^{\beta}_{\dot{\alpha}}} & = \ q\sum\limits_{n=0}^{\infty}[\gamma^{\mu}\gamma^{\nu}\psi (q\textit{\textbf{x}}^{\beta}_{\dot{\alpha}})]^{n},\\
\label{i2}\int_{\Omega_{q}}\frac{\psi [(1-q^{-1})\textit{\textbf{v}}_{\dot{k}}]\mathbf{d}_{q}\textit{\textbf{v}}_{\dot{k}}}{\textit{\textbf{v}}_{\dot{k}}-\textit{\textbf{x}}^{\beta}_{\dot{\alpha}}} & = \ \sum\limits_{n=0}^{\infty}[\gamma^{\mu}\gamma^{\nu}\psi [(1-q^{-1})\textit{\textbf{x}}^{\beta}_{\dot{\alpha}}]]^{n},\\
\label{i3}\int_{\Omega_{q}}\frac{\psi (q\textit{\textbf{p}}_{\dot{k}})\mathbf{d}_{q}\textit{\textbf{p}}_{\dot{k}}}{\textit{\textbf{p}}_{\dot{k}}-\textit{\textbf{u}}^{\beta}_{\dot{\alpha}}} & = \ q\sum\limits_{n=0}^{\infty}[\gamma^{\mu}\gamma^{\nu}\psi (q\textit{\textbf{u}}^{\beta}_{\dot{\alpha}})]^{n},\\
\label{i4}\int_{\Omega_{q}}\frac{\psi [(1-q^{-1})\textit{\textbf{p}}_{\dot{k}}]\mathbf{d}_{q}\textit{\textbf{p}}_{\dot{k}}}{\textit{\textbf{p}}_{\dot{k}}-\textit{\textbf{u}}^{\beta}_{\dot{\alpha}}} & = \ \sum\limits_{n=0}^{\infty}[\gamma^{\mu}\gamma^{\nu}\psi [(1-q^{-1})\textit{\textbf{u}}^{\beta}_{\dot{\alpha}}]]^{n}.
\end{align}

The formulation of Dirac operators in $q$-deformed spinorial coordinates constitutes a significant development in the study of quantum theories defined on noncommutative spaces. These structures enable a rigorous extension of the foundations of quantum mechanics and quantum field theory to contexts where classical continuous symmetries are replaced by quantum symmetries, described in terms of Hopf algebras and quantum groups.

Even in its classical form, the Dirac operator plays a central role in the description of spin-$\frac{1}{2}$ particles, such as electrons, and provides a relativistic formulation of quantum mechanics. Its generalisation to the $q$-deformed framework incorporates effects arising from discretisation, quantum curvature, and underlying noncommutative geometries. This generalisation is particularly relevant in theoretical approaches that seek to unify quantum mechanics with gravitational phenomena.

From an algebraic perspective, $q$-deformed Dirac operators are intimately connected with representations of $q$-deformed Clifford algebras and with noncommutative geometries in the sense of Connes. Their study permits the characterisation of quantum manifolds equipped with spinorial structure, thereby paving the way for more general formulations of field theory on curved or quantum spacetimes.

In physical terms, these operators facilitate the modelling of fermionic dynamics in the presence of $q$-deformed external fields, including noncommutative electromagnetic potentials and generalised gauge fields. Notably, the inclusion of a $q$-deformed Maxwell algebra introduces novel forms of interaction between fields and particles, with potential applications to effective models of quantum gravity, quantum cosmology, and condensed matter systems exhibiting discrete or topological symmetries.

The most substantial contribution lies in the capacity of these operators to encapsulate, within a single mathematical framework, the interplay of quantum symmetries, generalised spinorial structures, and noncommutative dynamics. As such, they offer a powerful and versatile tool for probing novel physical regimes in which the geometry of spacetime itself is governed by quantum principles.

\bibliographystyle{plain}       
\bibliography{myBiblib}
 ------------------------------------------------------------------------
\end{document}